\documentclass[sigconf,10pt]{acmart}
\usepackage{fancyhdr}
 \usepackage{amsmath}
\usepackage{cases}
\usepackage{stfloats}
\usepackage{enumerate}
\usepackage{url}
\usepackage{graphicx}
\usepackage{balance}
\usepackage{algorithm}
\usepackage{algpseudocode}

 \usepackage{etoolbox}

\newtoggle{ACM}
\toggletrue{ACM}

\newtoggle{IEEEcls}
\toggletrue{IEEEcls}
\togglefalse{IEEEcls}

\newtoggle{FullVersion}
\toggletrue{FullVersion}
\togglefalse{FullVersion}

\floatname{algorithm}{Algorithm}

\iftoggle{IEEEcls}{
\usepackage{amsthm}
\theoremstyle{plain}
\newtheorem{theorem}{Theorem}

\theoremstyle{definition}
\newtheorem{definition}{Definition}
 
\newtheorem{lemma}{Lemma}

}

\newif\ifNotUse  
 \NotUsetrue

 \iftoggle{ACM}{
\setcopyright{none}  
\acmDOI{}
\acmPrice{}
\acmISBN{}
\acmConference[ ]{ }{ }{ }
\acmYear{2018}

}

\begin{document}
\title{Memory efficient distributed sliding super point cardinality estimation by GPU}
\author{Jie Xu}
\affiliation{%
  \institution{School of computer science and engineer, Southeast university}
  \city{Nanjing}
  \state{China}
 }
\email{xujieip@163.com}

\iftoggle{IEEEcls}{
\author{\IEEEauthorblockN{Jie Xu
                          
                                  }
\IEEEauthorblockA{School of Computer Science and Engineering\\ South East University \\
Nanjing, China\\
Email: xujieip@163.com}

}

}

\begin{abstract}
Super point is a kind of special host in the network which contacts with huge of other hosts. Estimating its cardinality, the number of other hosts contacting with it, plays important roles in network management. But all of existing works focus on discrete time window super point cardinality estimation which has great latency and ignores many measuring periods. Sliding time window measures super point cardinality in a finer granularity than that of discrete time window but also more complex. This paper firstly introduces an algorithm to estimate super point cardinality under sliding time window from distributed edge routers. This algorithm's ability of sliding super point cardinality estimating comes from a novel method proposed in this paper which can record the time that a host appears. Based on this method, two sliding cardinality estimators, sliding rough estimator and sliding linear estimator, are devised for super points detection and their cardinalities estimation separately. When using these two estimators together, the algorithm consumes the smallest memory with the highest accuracy. This sliding super point cardinality algorithm can be deployed in distributed environment and acquire the global super points' cardinality by merging estimators of distributed nodes. Both of these estimators could process packets parallel which makes it becom possible to deal with high speed network in real time by GPU. Experiments on a real world traffic show that this algorithm have the highest accuracy and the smallest memory comparing with others when running under discrete time window. Under sliding time window, this algorithm also has the same performance as under discrete time window. 
\end{abstract} 
 
 
  \iftoggle{IEEEcls}{
}

\maketitle
 
\section{Introduction}
In nowadays network, there are huge packets passing through Internet every second \cite{cisco:NetForcast}. It is too expensive to measure every host in the network. An efficient way is to focus on special ones which have great influence on the network security and management. The super point, a host which communicates with lots of others, is one of such special hosts
, such as Web servers, P2P spreaders, DDoS victims, scanners and so on. For a host, the number of other hosts communicating with it is called as its cardinality. Detecting super point and estimating its cardinality can help us with network management and security\cite{scan:surveyPortScansAndDetection}\cite{Scan:NetworkScanDetectionWithLQS}\cite{rtID:RealTimeIntrusionDetectionandPreventionSystem}. It is also a foundation module of many instruction detection system\cite{Secure:SnortLightweightIntrusionDetectionNetworks}.

For example, DDoS (Distributed Denial of Service) attack is a heavy threat to the Internet\cite{DosC:AnalysisSimulationDDOSAttackCloud}\cite{DDos:ATaxonomyOfDDoSAttack}. It appears at the beginning of the Internet and becomes complex with the rapid growth of the network technology.  Although many defense algorithms have been proposed, most of them are too elaborate to deploy on the high-speed network. The peculiarity of a victim under DDoS attack is that it will receive huge packets with different source IP addresses in a short period\cite{DDos:surveyDefenseFloodingAttacks}\cite{DDos:AttackingDDoSAtSource}. A DDoS victim is a typical super point\cite{DDos:surveyCoordinatedAttacksIntrusionDetection}  .  

Super point only accounts for a small fraction of the overall hosts. If we detect super points first and spend more monitoring resource to them, we can defense DDoS much more efficiently.  Real-time super points detection on core network is an important step of these applications. 

Many researchers try to use small and fast memory, such as static random accessing memory SRAM, to detect super point. These algorithms used estimating method to record hosts' cardinalities and restore super points at the end of a time period. But the accuracy of these algorithms will decrease with the reduction of memory. And their memory requirement grows rapidly with the number of the packets in a certain period.

At the same computing platform, estimation algorithms are faster than precise algorithm because hash table needs an additional operation to deal with collision problem. Most of the previous algorithms tried to accelerate the packets processing speed by using fast memory SRAM. But the small size SRAM limits the accuracy of these algorithms in a high-speed network. What's more, estimation algorithm requires lots of computation operations and the computation ability of CPU is also the bottleneck. Parallel computation ability of GPU (Graphic Processing Unit) is stronger than that of CPU because of its plenty operating cores. When using GPU to scan packets parallel, we would get high throughput. 

Super point detection has been researched for a long time because of its importance. And many excellent algorithms have been proposed recent years. But these algorithms only work for discrete time window, under which there is no duplicating time period between two adjacent windows. These algorithms will reinitialize at the beginning of every window and discard hosts' cardinality information in previous time. This time window splits host cardinality into discrete piece and can only report super points at the end of a window which has a latency of the size of time window. Sliding time window which moves a small unit smoothly has a better measurement result than discrete time window. It stores and updates host cardinality information incrementally. Sliding time window estimates super point cardinality more precisely because it is not affected by the starting of window. And sliding time window reports super point more timely for the sake that the moving unit time is much smaller than the size of discrete time window and at the end of each moving unit time, super point will be detected and reported once. But super point detection and cardinality estimation under sliding time window is more complex than that under discrete time window because it maintains hosts state for some previous time.

The speed of nowadays network is growing rapidly. For a core network, it always contains several border routers which locate at different places. How to detect overall super points and estimate their cardinalities from all of these distributed routers under sliding time window is more difficult than from a small single router.  

To overcome previous algorithms weakness, we devise a novel distributed super points algorithm which can detect super points and estimate their cardinalities under slidng time window. This algorithm also consumes the smallest memory. When running on a low-cost GPU, our algorithm can deal with core network traffic in real time. The contributions of this paper are list following:
\begin{enumerate}
\item A distributed sliding super point detection algorithm is proposed.

\item A memory efficient distributed sliding super point cardinality estimation algorithm is devised.

\item Deploy the sliding super point detection and cardinality estimation algorithm on a common GPU to deal with core network in real time.

\end{enumerate}

In the next section, we will introduce previous super point detection algorithm and analyze their merit and weakness. In section 3, a novel super point detection algorithm under sliding time window will be introduced together with how to run on GPU. Section 4 describes how to estimate host's cardinality under sliding time window and how to deploy it on GPU. Section 5 shows experiments of real world core network traffic. And we make a conclusion in the last section.

\section{Related Work}
High speed network super point detection has been researched for a long time. At first, sampling method was used to solve the problem of slow processing speed\cite{HSD:streamingAlgorithmFastDetectionSuperspreaders}\cite{HSD:identifyHighCardinalityHosts}. But sampling method affected the accuracy of these algorithm especially in the situation where a high sampling rate was adopted. Then many works tried to improve the processing speed by using high speed memory, such as CBF\cite{HSD:LineSpeedAccurateSuperspreaderIdentificationDynamicErrorCompensation}, DCDS\cite{HSD:ADataStreamingMethodMonitorHostConnectionDegreeHighSpeed}  , VBFA\cite{HSD:DetectionSuperpointsVectorBloomFilter}.

Chen et al. \cite{HSD:LineSpeedAccurateSuperspreaderIdentificationDynamicErrorCompensation} proposed a contacting hosts estimator called counter bloom filter CBF based on the theorem of bloom filter. When a flow appears, several counters in CBF were added by one. A flow only updated CBF once. This algorithm had a high accuracy and speed when running with a single thread on SRAM. According to the statement of the authors, this algorithm could scan 2 million packets per second. But this speed was still too low for nowadays high speed network which forwards more than 6 million packets every second. And this algorithm couldn't work on parallel and distributed environment because a flow may update CBF many times in these cases. 

Wang et al.\cite{HSD:ADataStreamingMethodMonitorHostConnectionDegreeHighSpeed} used linear estimator \cite{DC:aLinearTimeProbabilisticCountingDatabaseApp} to estimate hosts' cardinalities and proposed a novel structure called DCDS based on Chinese Remainder Theory(CRT) which can restore hosts directly. But CRT is so complex that it requires many computing resource and time. To overcome this weakness, Liu et al.\cite{HSD:DetectionSuperpointsVectorBloomFilter} proposed a structure called VBF which was similar to a bloom filter. VBF regained hosts by bits comparing and concatenation, instead of by CRT. VBF had a much faster speed than DCDS because of its simple regain procedure. VBF used sub bits of IP address to map a host to several linear estimators. Sub bits can be acquired quickly but had little randomness which caused that most of linear estimators in VBF were not be used and memory was wasted.

Those algorithms only focused on how to speed up by reducing memory latency. They neglected the huge computing resource requirement. GPU can solve this two problems, high memory operation speed and plenty computing resource, all together.  

GPU is the best desktop super computing platform which has the same computing ability as a small cluster. In a single GPU chip, hundreds or thousands of cores sharing a big global graphic memory. Different threads can read and store this memory parallel. Although a core in GPU is a little slower, lower frequency, than a core in CPU, the total computing ability of these hundreds of GPU cores is much stronger than that of a CPU which only have teens of cores at most. The convenient program environment, such as CUDA\cite{GPU:OptimizationPrinciplesApplicationPerformanceCUDA}, OpenCL\cite{OpenCL:AParallelProgrammingStandardForHeterogeneousComputingSystems}, makes GPU becomes one of the most popular parallel computing platform. 

GPU was firstly used to detect super points by Seon-Ho et al.\cite{HSD:GPU:2014:AGrandSpreadEstimatorUsingGPU}. They deployed a novel structure called virtual vector on GPU to estimate hosts' cardinalities. But virtual vector could only estimate contacting hosts number under discrete time window and super points could not be reconstructed from it directly.

None of these algorithms can estimate super points cardinalities under discrete time window. This paper will introduce a sliding super points detection algorithm and describe how to deploy it on GPU for real time distributed running.

\section{Sliding super point detection}
\subsection{Sliding super point}
Suppose there are two networks $ANet$ and $BNet$. These two networks are contacting with each other through a set of edge routers $ER$. $ANet$ might be a city-wide network or even a country-wide network. And $BNet$ might be another city-wide network or the Internet. All traffic between $ANet$ and $BNet$ could be observed from $ER$. Split this traffic by successive time slices as shown in figure \ref{SlidingWindowModal}.

\begin{figure}[!ht]
\centering
\includegraphics[width=0.47\textwidth]{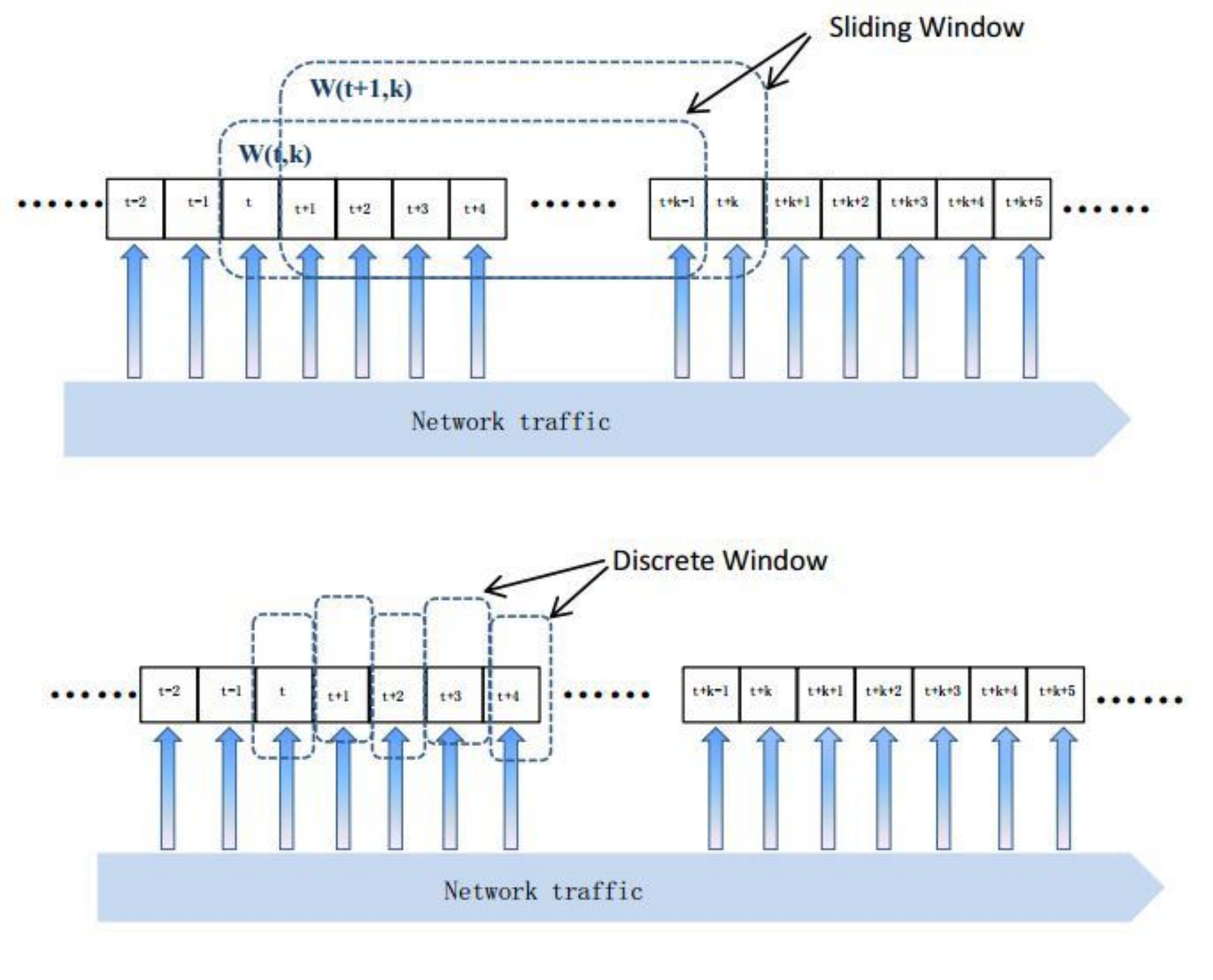}
\caption{Sliding time window and discrete time window}
\label{SlidingWindowModal}
\end{figure}

These time slices have the same duration. The length of a time slice could be 1 second, 1 minute or any period in different situations. Every time slice is identified by a number. A sliding time window $W(t,k)$ contains $k$ successive slices starting from the $t$ time slice as shown in the top part of figure \ref{SlidingWindowModal}. Sliding time window will move forward one slice once a time. So two adjacent sliding time windows contain $k-1$ same slices. When $k$ is set to 1, there is no duplicate time period between two adjacent windows, which is the case of discrete time window in the bottom part of figure \ref{SlidingWindowModal}. 

Let $ANet$ be the network from which we want to detect super points. A host's packets stream in a sliding time window is defined as below.

\begin{definition}[Packets stream of a host]
\label{def-slidingPktStream}
For a host $aip \in ANet$, every packet passing through $ER$ in sliding time window $W(t,k)$ which has $aip$ as source or destination address composes packets stream of $aip$, written as $Pkt(aip,t,k)$.
\end{definition}

$aip$'s opposite hosts stream $ST(aip,t,k)$ could be derived from $Pkt(aip,t,k)$ by extracting the other IP address except $aip$. A IP address $bip$ may appear several times in $ST(aip,t,k)$ because $aip$ can send several packets to $bip$ or receive many packets from $bip$. Hosts in $ST(aip,t,k)$ make up of opposite hosts set of $aip$, written as $OH(aip,t,k)$. The number of element in $OH(aip,t,k)$, denoted as $|OH(aip,t,k)$, is no bigger than that of $ST(aip,t,k)$. $|OH(aip,t,k)|$ is the cardinality of $aip$ in sliding time window $W(t,k)$. Sliding super point is defined according to host's cardinality.

\begin{definition}[Sliding super point]
\label{def-slidingSuperPoint}
For a host $aip \in ANet$, if $|OH(aip,t,k)| \geq \theta$, $aip$ is a sliding super point in sliding time window $W(t,k)$. Where $\theta$ is a positive integer.
\end{definition}

When $k=1$, sliding super point could be called as super point. Threshold $\theta$ is defined by users for different applications. It could be selected according to the average cardinality of all host in the past or the normal cardinality of a server. How to get $|OH(aip,t,k)|$ from $ST(aip,t,k)$ is a hard task. Because packets pass through $ER$ with high speed and every packet could only be scanned a time in the stream. How to process every coming packet and give an accurate estimation of $|OH(aip,t,k)|$ at the end of the last time slice of $W(t,k)$ is the key step in the whole algorithm.
\subsection{Detecting sliding super point with rough estimator}
The key step of sliding super point detecting is to determine if a host is a super point in a sliding time window. Because discrete time window is a special case of sliding time window and working under discrete time window is much simpler than that under sliding time window, we firstly introduce how to judge super point under discrete time window and then give its sliding time window version. 

For a host $aip$, the task of judging super point under discrete time window is to determine if $|OH(aip,t,1)| \geq \theta$ by scanning every host in $ST(aip,t,1)$ once. Rough estimator $RE$ proposed in this paper is a memory efficient algorithm which can tell if a host is a super point in a time period with only 8 bits. These 8 bits are initialized to zero at the begin of a time period. $RE$ samples and records hosts in $ST(aip,t,1)$ by the least significant bits of their hashed values. Least significant bit of an integer is defined in the below.

\begin{definition}[Least significant bit, LSB]
\label{def-LeastSignifcantBit}
Given an integer $i$, let $BIN(i)$ represent its binary formatter. The least significant bit of $i$, $LSB(i)$, is the index of the first `1' bit of $BIN(i)$ starting from right.
\end{definition}

For example, $LSB(3)=0$, $LSB(40)=3$. The binary formatters of 3 and 40 are "11" and "101000". The first bit of $BIN(3)$ is `1', so $LSB(3)$ equals to 0. While $BIN(40)$ meets its first `1' until the fourth bit, so its $LSB$ is 3. For every host $bip$ in $OH(aip,t,1)$, RE hashes it to a random value between 0 and $2^{32} -1$ by a hash function $H_1$ \cite{hash_AsmallApproximatelyMinWiseIndependentHF}. If $LSB(H_1(bip))$ is smaller than an integer $\tau$, this IP will not be recorded by $RE$ where $\tau$ is derived from $\theta$ by equation \ref{eq_getLsbThreshold}.

\begin{equation}\label{eq_getLsbThreshold}
\tau=ceil(log_2(\theta/ \eta ))
\end{equation} 

When  $LSB(H_1(bip)) \geq \tau$, a bit selected by $H_2(bip)$ will be set where $H_2$ is another hash function mapping $bip$ to a value between 0 and $\eta -1$.

$RE$ deals with every host in $ST(aip,t,1)$ in this way. At the end of slice $t$, if the number of `1' bits is no smaller than $\rho * \eta$, $|OH(aip,t,1)|$ is judged as bigger than $\theta$ by $RE$, where $\rho = 0.99 * (1-e^{-1/3})$. $\rho$ is acquired from \cite{DC:AnOptimalAlgorithmDistinctElementProblem}.

$RE$ has a high probability to report a super point. 

\iftoggle{FullVersion}{
Then we will give its mathematical analyze.

\begin{lemma}
\label{la-RE_fullN}
Suppose there are $\alpha$ different balls, $\eta$ different boxes and $\alpha \geq \eta$. Throw all of these balls randomly to these boxes. Let $FN(\alpha, \eta)$ represent the number of situations that every $\eta$ boxes has at least a ball. Then $FN(\alpha, \eta)= \eta ^ \alpha -\sum_{i=1}^{r-1}C_r^i*FN(\alpha, i)$ and $FN(\alpha,1) =1$.
\end{lemma} 
\begin{proof}
There are total $\eta ^ \alpha$ situations to throw $\alpha$ balls to $\eta$ boxes. When there is only a box, there is only a situation, throwing all balls to it. When throwing all balls to $i$ boxes and all of these boxes contain at least one ball, there are $C_r^i*FN(\alpha, i)$ situations. Deduct all situations that all balls are thrown to a subset of $\eta$ boxes from $\eta ^ \alpha$, the rest is the number of situations that there are no empty boxes.
\end{proof} 

\begin{theorem}
 \label{th-RE_NoneEmptyBoxesN}
Throw $\alpha$ balls to $\eta$ boxes. Let $\eta_1$ represent the number of boxes that contain at least a ball. The number of situations that there are $\eta_1$ balls are none empty is $FN(\alpha,\eta,\eta_1)=C_{\eta}^{\eta_1}*FN(\alpha,n)$, where $1\leq n \leq \eta$.  
\end{theorem}
\begin{proof}
The rest $\eta-\eta_1$ balls are empty. There are $C_{\eta}^{{\eta}_1}$ situations to choose $\eta - \eta_1$ empty balls. Each situation has $FN(\alpha,\eta_1)$ methods to throw $\alpha$ balls. So the number of total situations is $C_{\eta}^{\eta_1}*FN(\alpha,\eta_1)$.
\end{proof}

$OH(aip,t,1)$ could be regarded as the set of balls and $\eta$ bits could be regarded as boxes in theorem \ref{th-RE_NoneEmptyBoxesN}. $\eta_1$ means the number of bits that are set to 1. 
}
Suppose there are $\alpha$ hosts in $OH(aip,t,1)$ updating $RE$. The probability that there are $\eta_1$ bits are set to 1 is :

\begin{equation}\label{eqt_prb_alphaNoneEmpytBox}
 Pr\{\alpha,\eta,\eta_1\}=\frac{FN(\alpha,\eta,\eta_1)}{{\eta}^{\alpha}}
\end{equation}

Every host in $OH(aip,t,1)$ has probability $\frac{1}{2^\tau}$ to update $RE$. So the probability that there are $\alpha$ hosts in $OH(aip,t,1)$ updating $RE$ is:

\begin{equation}\label{eqt_prb_alphaNum}
\begin{split}
& Pr\{|OH(aip,t,1)|,\alpha\} \\
& =C_{|OH(aip,t,1)|}^{\alpha}*({\frac{1}{2^ \tau}})^\alpha*(1-\frac{1}{2^\tau})^{|OH(aip,t,1)|-\alpha}
\end{split}
\end{equation}

Combine equation \ref{eqt_prb_alphaNoneEmpytBox} and \ref{eqt_prb_alphaNum}, we will get the probability that there are $\eta_1$ `1' bits in $RE$ after scanning $ST(aip,t,1)$ as shown in equation \ref{eqt_prb_1BitsNAfterScanningHostStrm}.
\begin{equation}
\label{eqt_prb_1BitsNAfterScanningHostStrm}
\begin{split}
& Pr\{|OH(aip,t,1)|,\eta,\tau,\eta_1\} \\
& =\sum_{\alpha=\eta_1}^{|OH(aip,t,1)|}Pr\{|OH(aip,t,1)|,\alpha\}*Pr\{\alpha,\eta,\eta_1\}
\end{split}
\end{equation}.

The probability that there are more than ${\eta}_1$ `1' bits in $RE$ after scanning $ST(aip,t,1)$ could be derived from \ref{eqt_prb_1BitsNAfterScanningHostStrm} as shown in equation \ref{eqt_prb_1BitsNPlusAfterScanningHostStrm}.
\begin{equation}
\label{eqt_prb_1BitsNPlusAfterScanningHostStrm} 
\begin{split}
& Pr^+\{|OH(aip,t,1)|,\eta,\tau,\eta_1\} \\
& =\sum_{\eta_1=n}^{\eta}Pr\{|OH(aip,t,1)|,\eta,\tau,\eta_1\}
\end{split}
\end{equation}

Equation \ref{eqt_prb_1BitsNPlusAfterScanningHostStrm} proofs that $RE$ has a high probability to detect super point in discrete time window. There is only one time slice in $W(t,1)$. So one bit is enough to represent if a host appears in this slice. But when $k>1$, a single bit can't know if some hosts still appear in $W(t+1,k)$ when the window sliding from $W(t,k)$. 

Unlike discrete version using $\eta$ bits, $RE$ sliding version $SRE$ uses $\eta$ short integers, where each short integer is unsigned short integer with 16 bits, to record host's cardinality. $SRE$ has five operations: initialization, hosts scanning, super point detection, slice updating, $SRE$s merging.

Before $SRE$ launching, every short integer will be initialized to the biggest value 65535. Like $RE$, $SRE$ scans every host in a time slice $t+k-1$ by hashing the host, comparing the LSB of hashed value with $\tau$. When the hashed value is no smaller than $\tau$, an short integer will be selected and set to 0, unlike $RE$ which sets a bit to 1. Let $SRE[i]$ point to the $i$th short integer in $SRE$. The weight of $RE$ is the number of 1 bits in it. But the weight of $SRE$ is the number of integer whose value is smaller than $k$, denoted as $|SRE|^k$. $SRE$ will check $aip$ at the end of time slice $t+k-1$. It judges if $aip$ is super point by comparing its weight $|SRE|^k$ with $\eta * \rho$ like $RE$. After checking $aip$, the window will slide to $W(t+1, k)$. $SRE$ will not initialize its short integers at the beginning of a new time slice. But every short integer will increment 1 before scanning new hosts in this time slice. A short integer in $SRE$ records the distance between the nearest time slice when it is set and the scanning time slice. So when the window moves forward, the distance will grow too. If a short integer is mapped by some host in the time slice, the distance is 0. When merging two $SRE$s together, short integers in the new $SRE$ will be selected from the biggest ones. Algorithm \ref{alg_SRE_ScanningHost} and \ref{alg_SREmerge} describe the hosts scanning and $SRE$ merging process.

\begin{algorithm}                       
\caption{Scan hosts}          
\label{alg_SRE_ScanningHost}                            
\begin{algorithmic}                    
\Require {\\opposite host $bip$,\\
Sliding rough estimator $SRE$\\} 
\State  $hbip<-H_1(bip)$
\State $lsb<-LSB(hbip)$
\If{$lsb < \tau$}
\State Return
\EndIf
\State $sInt<-H_2(bip)$
\State $SRE[sInt]<-0$
\end{algorithmic}
\end{algorithm}

\begin{algorithm}                       
\caption{Merge sliding rough estimators}          
\label{alg_SREmerge}                            
\begin{algorithmic}                    
\Require {\\sliding rough estimator $SRE_1$, $SRE_2$} 
\Ensure{merged sliding rough estimator $SRE_3$ \\}  

\State Init $SRE_3$
\For {$sInt \in [0,\eta-1]$}
\If{$SRE_1[sInt]>SRE_2[sInt]$}
\State $SRE_3[sInt]<-SRE_1[sInt]$
\Else
\State $SRE_3[sInt]<-SRE_1[sInt]$
\EndIf
\EndFor
\State Return $SRE_3$
\end{algorithmic}
\end{algorithm}

There is no accessing conflict in $SRE$ which means $SRE$ could deal with several hosts parallel to speed up the process. $SRE$ can determine if a host is super point in a window. But there are millions of hosts in $ANet$ for a high speed network. It's too expensive and slow to allocate a $SRE$ for every host in $ANet$. A smart structure based on $SRE$ is devised to solve this problem. The new algorithm uses a fixed number of $SRE$s to record hosts and restore sliding super points at the end of a time slice as shown in the next part.

\subsection{Running rough estimator on GPU}
Many hosts in $ANet$ have a smaller cardinality. Allocating a $SRE$ for every host will waste lots of memory and slow down hosts scanning speed. Based on $SRE$, a smart structure, called as reversible sliding rough estimator array RSRA, is proposed. RSRA contains $2^q$ columns and r rows of sliding rough estimators as shown in figure \ref{fig_RSRA}. Let RSRA[i,j] point to the $SRE$ in the $i$th row, $j$th column. This structure is reversible because sliding super point could be reconstructed from it without any other data. This reversible ability comes from a novel hash functions group, reversible hash functions group RHFG. 
\begin{figure}[!ht]
\centering
\includegraphics[width=0.47\textwidth]{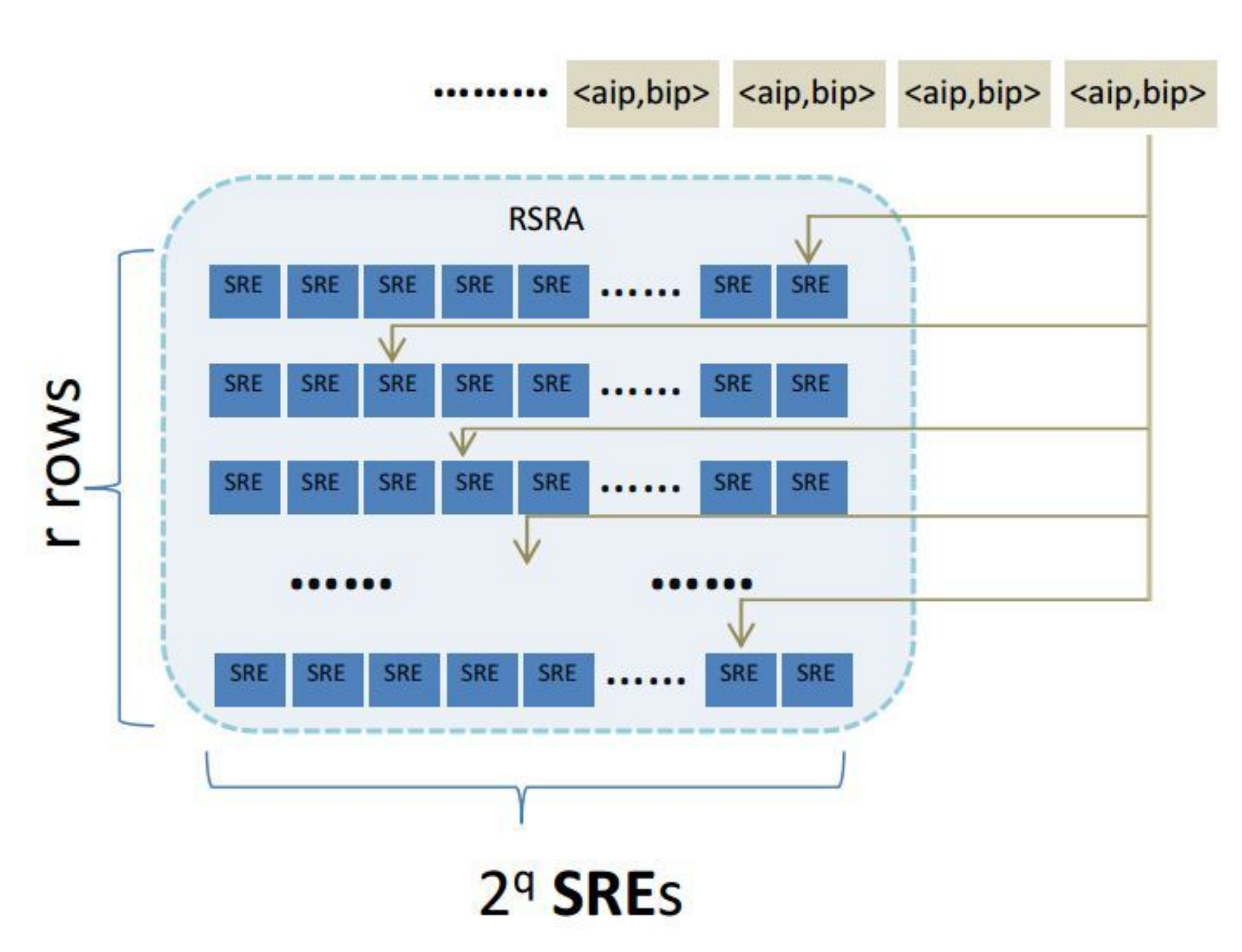}
\caption{Reversible sliding rough estimator array}
\label{fig_RSRA}
\end{figure}

RHFG is an array of r hash functions, each of which hashes an IP address to a value between 0 and $2^q-1$. Let RHFG[i] represent the $i$th hash function. RHFG[0] is a random hash function \cite{hash_ThePowerOfSimpleTabulationHashing} which maps a IP address to an integer between 0 and $2^q-1$, $0 \leq RHFG[0](aip) \leq 2^q-1$ where $aip \in ANet$. The rest $r-1$ hash functions are derived from $RHFG[0]$ according the following equation.
\begin{equation}\label{eqt_RHFG_restHashes}
\begin{split}
& RHFG[i](aip)= \\
& \ \ \ \ \ ((aip>> (i*\delta) )\ XOR\ RHFG[0](aip) ) mod (2^q)
\end{split}
\end{equation}

In equation \ref{eqt_RHFG_restHashes}, $1 \leq i \leq r-1$, $aip \in ANet$. $\delta$ is an positive integer that smaller than q and $(r-2)*\delta+q \geq 32$. ``XOR" is the bit wise exclusive or operation. "$>>$" is the bit wise right shift operation. According to the property of ``XOR", we can recover $(aip>> (i*\delta)) mod (2^q)$ by equation \ref{eqt_RHFG_bits_from_xor_cid}.

\begin{equation}\label{eqt_RHFG_bits_from_xor_cid}
\begin{split}
& (aip>> (i*\delta) ) mod (2^q)= \\
&  \ \ \ \ \ RHFG[i](aip)\ XOR\ RHFG[0](aip)  
\end{split}
\end{equation}

$(aip>> (i*\delta) ) mod (2^q)$ is q successive bits of $aip$ starting from $i*\delta$, written as $B(i)$. Because $(r-2)*\delta+q \geq 32$, every bit in $aip$ will appear in some $B(i)$ where $1 \leq i \leq r-1$. For a host $aip \in ANet$, let RHFG(aip) represent the array of r hashed values where $RHFG(aip)=\{RHFG[0](aip),RHFG[1](aip)$ $,\cdots,RHFG[r-1](aip)\}$. $aip$ could be regained from $RHFG(aip)$ by extracting bits from $B(i)$.

$RHFG$ has high randomness and reversible ability. It is used to select $r$ sliding rough estimators from each row of RSRA for every host in $ANet$. For a host $aip$ in $ANet$, its r sliding rough estimators are denoted as $RSRA(aip)=\{RSRA[0,RHFG[0](aip)]$ $, RSRA[1,RHFG[1](aip)]$ $,\cdots,$ $RSRA[r-1,RHFG[r-1]$ $(aip)]\}$. A IP pair is a set of two IP addresses extracting from a packet where one address is in $ANet$ and the other is in $BNet$. When a IP pair comes, these r sliding rough estimators will be updated at the same time as shown in algorithm \ref{alg_UpdateRSRA_oneIPpair}.

\begin{algorithm}                       
\caption{Update RSRA}          
\label{alg_UpdateRSRA_oneIPpair}                            
\begin{algorithmic}                    
\Require {\\IP pair $<aip,bip>$,\\
Reversible hash functions group $RHFG$,\\
Reversible sliding rough estimator array $RSRA$} 
\State  $hbip<-H_1(bip)$
\State $lsb<-LSB(hbip)$
\If{$lsb < \tau$}
\State Return
\EndIf
\State $DRidx \Leftarrow H_2(bip)$  
\For{$i \in [0,r -1]$}
\State $COLidx \Leftarrow RHFG[i](aip)$ 
\State $sre \Leftarrow RSRA[i, COLidx]$
\State $sre[DRidx] \Leftarrow 0$ \label{alg-line-updateRSRA_SEupdate}
\EndFor 
\end{algorithmic}
\end{algorithm}
Algorithm \ref{alg_UpdateRSRA_oneIPpair} describes how to update RSRA for a IP pair. But there are millions of IP pairs every second for example in a 40 Gb/s network. Dealing with these IP pairs one by one will consume much time for a single thread. In algorithm \ref{alg_UpdateRSRA_oneIPpair} only line \ref{alg-line-updateRSRA_SEupdate} update memory while others are computing operations such as getting sliding estimator index in RSRA, calculating which distance recorder to be set. 
A distance recorder could be set to zero multi times which makes sure that there is no need to synchronize among memory access and several IP pairs could update RSRA at the same time. 

Nowadays CPU contains several cores, from 2 to 22 or more such as Intel E5-2699v4. When exploiting all cores of CPU to scan IP pairs parallel, the processing speed will be raised. But the memory bandwidth of CPU will limit the increment. What's more, the price of CPU grows rapidly with the number of cores because the single core of CPU is so powerful, high frequency and complex control ability, that it occupies much space on chip. 

Unlike CUP's core, each core of GPU is a little simple, lower frequency and fewer controlling unit, but occupies much smaller space. So a GPU could contain hundreds or even thousands of cores in a chip easily.  The total computation ability of GPU is much stronger than that of CPU. And GPU has a lower memory access latency because it has several memory controllers for multi threads. For tasks, which dealing with different data by the same instructions, GPU can acquire a high speed-up. IP pair scanning is such a task.

IP pair scanning consumes the most time in sliding super points detection because the huge number of IP pairs appearing in every slot. Every IP pair is processed by the same algorithm, algorithm \ref{alg_UpdateRSRA_oneIPpair}. So thousands of threads running algorithm \ref{alg_UpdateRSRA_oneIPpair} could be launched to scan thousands of IP pairs at the same time. Figure \ref{GPU_IPpairScan} illustrates how to detect sliding super point on GPU.

\begin{figure}[!ht]
\centering
\includegraphics[width=0.47\textwidth]{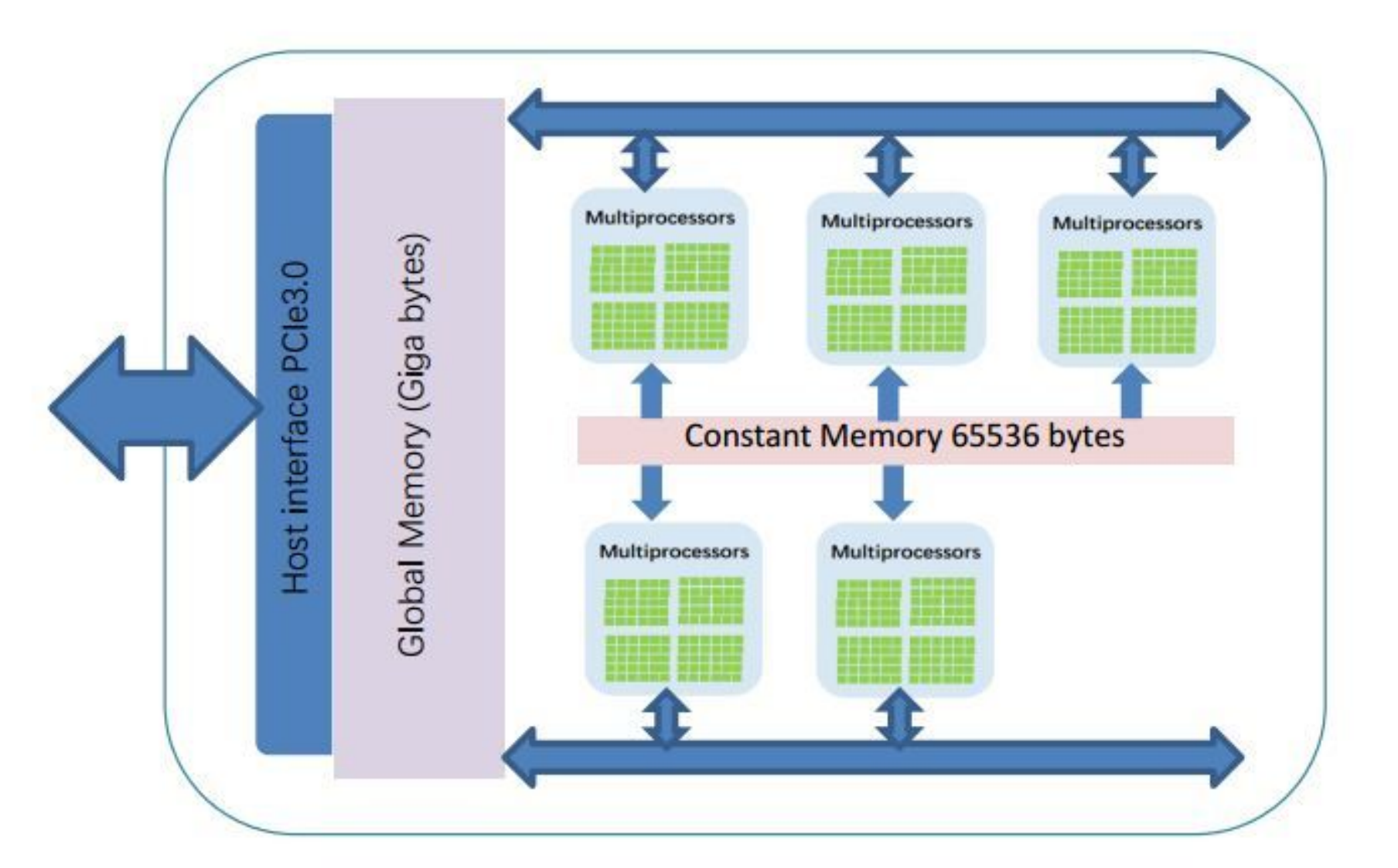}
\caption{Hosts scan on GPU}
\label{GPU_IPpairScan}
\end{figure}

IP pair will be copied to GPU's global memory by PCI-E bus.  A IP pair buffer on GPU memory, which can contain $\alpha$ IP pairs, is allocated to receiving IP pairs. When the buffer is full or IP pairs in a slot are all copied, the same number of threads, as the number of receiving IP pairs, will be launched on GPU to process these IP pairs. Every thread reads one IP pair from global memory and update a distance recorder in RSRA which locating in global memory too. For IP v4 address, the buffer of IP pair occupies $8*\alpha$ bytes. When $\alpha$ is set to $2^{15}$, this buffer needs 256 KB memory. The graphic memory on GPU, ranging from 1 GB to 11 GB, is big enough to hold it. Although the RSRA requires more memory than IP pair buffer, the global memory is plenty enough to store a RSRA which is big enough for a 40 Gb/s networks. Other running parameters, such as hash function parameters, r, q and $\delta$, are stored in the constant memory which is read only but has high speed. A low cost GPU, which can be brought within 200 dollars, is fast enough to scan IP pairs in 40 Gb/s in real time. RSRA could be deployed in many nodes at different places to processes packets in the distributing environment. Every node should maintain the same size $RSRA$ and same hash functions' parameters.

After scanning all IP pairs in a slice, sliding super point will be reconstructed from RSRA. If there are many nodes in the distributing environment, $RSRA$ in these nodes should be merged to a global one by $SRE$ merging and sliding super point will be detected from this global $RSRA$. According to the feature of RHFG, if $RSRA(aip)$ is known, $aip$ could be restored from it. But $RSRA(aip)$ is not stored directly. According to the definition, if $aip$ is a sliding super point, every sliding rough estimator in $RSRA(aip)$ will contain no less than $\eta*\rho$ short integers whose value is smaller than $k$. The sliding rough estimator whose weight is no less than $\eta * \rho$ is called as hot sliding rough estimator denoted by $HSE$.  

A candidate tuple $CT$ consists of r $HSE$s could be acquired by selecting a $HSE$ from every $HSE(i)$ where $ 0 \leq i \leq r-1$. $CT=\{he_0,he_1,he_2,\cdots,he_{r-1}\}$ where $he_i \in HE(i)$. Sliding super point could be regained by testing all of these candidate tuples. 

Two additional buffers of candidate tuples are used in this algorithm, one for storing and the other for reading. Their roles exchange in different levels when adding hot estimators in different rows. Let $SCTB$ point to the candidate tuple buffer for storing and $RCTB$ point to the candidate tuple buffer for reading.

Candidate tuple in these two buffers grows incrementally from empty to a valid tuple containing $r$ hot estimators in different rows. Let $CTB_1$ and $CTB_2$ represent these two buffers respectively. Figure \ref{GPU_candidateTupleGrow} shows how candidate tuple grows with two buffers' support.

\begin{figure}[!ht]
\centering
\includegraphics[width=0.47\textwidth]{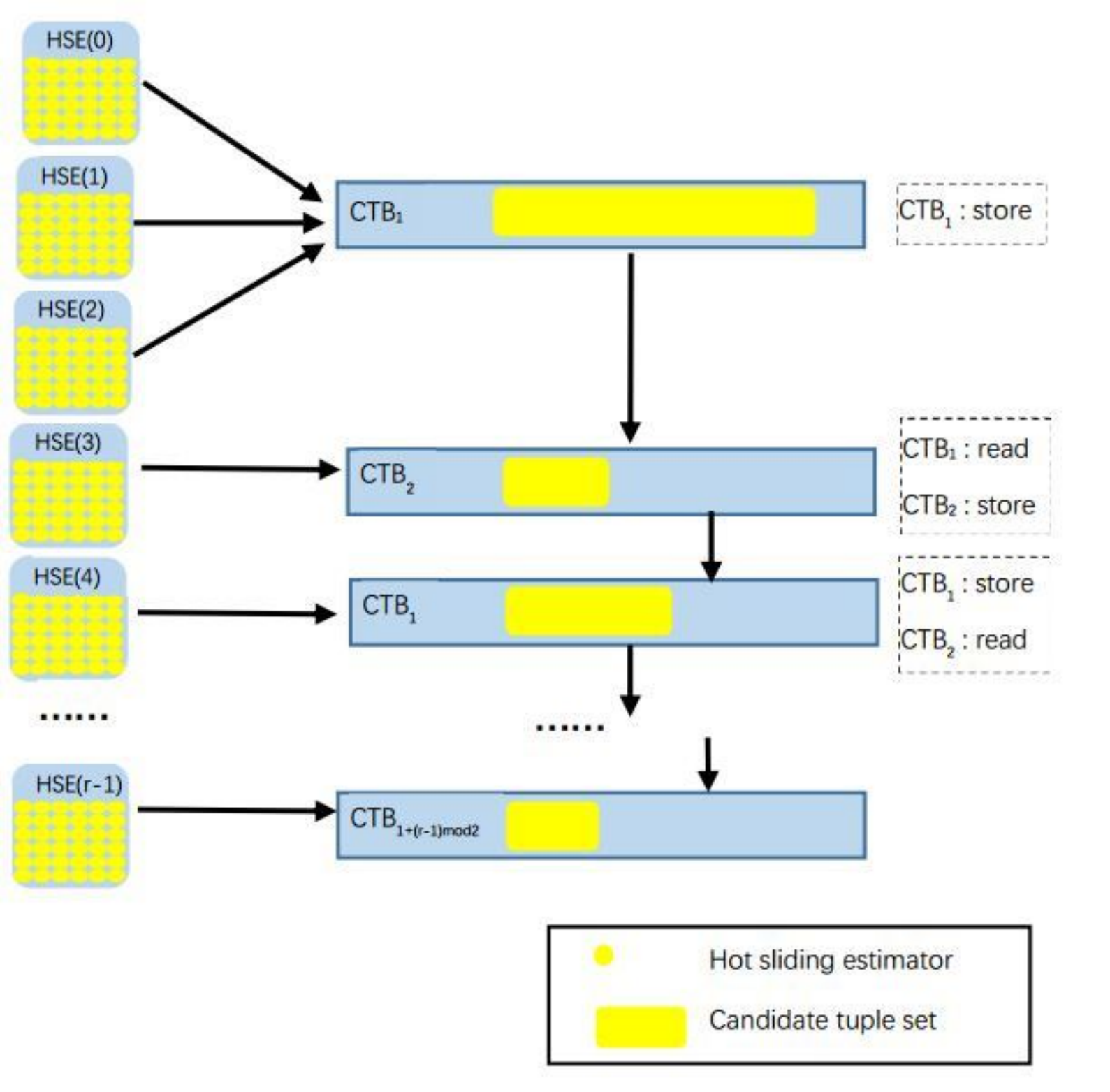}
\caption{Regain sliding super points on GPU}
\label{GPU_candidateTupleGrow}
\end{figure}

These two candidate tuple buffers are located on GPU's global memory. Candidate tuple $CT_2=\{he_0,he_1,he_2\}$ with three hot estimators, selected from $HSE(0)$, $HSE(1)$, $HSE(2)$ separately, will be inserted into $CTB_1$ after checking. The checking procedure is to test if $B(1)$ and $B(2)$ extracted from $CT_2$ is valid, if the left $q-\delta$ bits of $B(1)$ is the same as the right $q-\delta$ bits of $B(2)$. Only when passing the test, will $CT_2$ be added to $CTB_1$.  When $q-\delta$ is big, only a small part of such candidate tuple will appear in $CTB_1$. The memory updating latency caused by candidate tuple insertion will be concealed by the huge parallel running threads on GPU. So candidate tuple checking determines the time consumption of a thread. When every thread deals with the same amount of candidate tuples, they will finish approximately at the same time. In this situation, the load of every thread is balance and GPU realizes its full potential.

There are total $Q=|HSE(0)|* |HSE(0)|* |HSE(0)|$ candidate tuples like $CT_2$. Suppose $V$ threads are launched on GPU to deal with these candidate tuples. Let $U$, $V$, $Q$ and $W$ be non-negative integers. In order to let every thread has the same candidate tuples to check, each thread will be assigned at least $U=Q/V$ candidate tuples evenly. Still there are $W=Q\ mod\ V$ candidate tuples rest. In these $V$ threads, every of first $W$ threads has $U+1$ candidate tuple and every of the rest $V-W$ threads has $U$ candidate tuple. Let $CT_2(i)$ represent the set of candidate tuples to be tested by the $i$th threads in GPU which can be acquired from $HSE(0)$, $HSE(1)$, $HSE(2)$. Algorithm \ref{alg_candidateBuf_H012} shows how every thread checks candidate tuples. 

\begin{algorithm}                       
\caption{global function on GPU\\
Generate candidate tuple $CT_2$}          
\label{alg_candidateBuf_H012}
\begin{algorithmic}                    
\Require {\\  
Hot sliding estimator set $HSE(0)$, $HSE(1)$, $HSE(2)$, \\
Storing candidate tuple buffer $SCTB$
\\} 
 
\State $TID \Leftarrow $ thread index 
\State get candidate tuple set from $HSE(0)$, $HSE(1)$, $HSE(2)$
\State $CT_2(TID) \Leftarrow$ candidate tuple set to be tested by this thread 
\For{$ct=\{he_0,he_1,he_2\} \in CT_2(TID) $ }
\State $B(1)\ \Leftarrow\ he_0\ XOR\ he_1$    \label{alg_line_gpu_h012_check_start}
\State $B(2)\ \Leftarrow\ he_0\ XOR\ he_2$
\If{ left $q-\delta$ bits of $B(0)$ not equal to right $q-\delta$ bits of $B(1)$} \label{alg_line_gpu_h012_check_end}
\State Continue
\EndIf                                     
 \State insert $ct$ into $SCTB$
\EndFor
\end{algorithmic}
\end{algorithm}

$CT_2(TID)$ could be acquired from $HSE(0)$, $HSE(1)$, $HSE(2)$ according to the index of a GPU thread. When testing candidate tuple in $CTB_2(TID)$, valid candidate tuple which passes checking process from line \ref{alg_line_gpu_h012_check_start} to \ref{alg_line_gpu_h012_check_end} will be stored in $SCTB$ for further checking with hot estimators in other rows.

When all threads finished, $CTB_1$ which has stored all valid candidate tuples extracting from the first three rows will work as reading buffer and the other buffer, $CTB_2$ will be used for storing new candidate tuple as shown in figure \ref{GPU_candidateTupleGrow}.

For $HSE(i)$ where $i \geq 3$, a new candidate tuple for checking is generated from a candidate tuple in reading tuple buffer, candidate tuple buffer which has stored valid candidate tuple, and a hot estimator in it. Then $Q=|RCTB|*|HSE(i)|$ where $|RCTB|$ means the number of candidate tuple storing in reading candidate tuple buffer. When $i$ is an odd number, $RCTB$ points to $CTB_1$, $SCTB$ points to $CTB_2$; when $i$ is an even number, $CBT_1$ and $CBT_2$ exchange roles. A new candidate tuple consists of a hot estimator in $HSE(i)$ and a candidate tuple in $RCTB$. The set of such new candidate tuple to be checked by the $j$th thread, $CT_i(j)$, could be generated from $HSE(i)$ and $RCTB$. Algorithm \ref{alg_updateCandidateTuple} shows how to check new candidate tuples.

\begin{algorithm}                       
\caption{global function on GPU\\
Update candidate tuple}          
\label{alg_updateCandidateTuple}
\begin{algorithmic}                    
\Require {\\  
Row index $i$,\\
Hot estimators set $HSE(i)$,\\
Storing candidate tuple buffer $SCTB$,\\
Reading candidate tuple buffer $RCTB$
\\}  
\State $TID \Leftarrow $ thread index 
\State $CT_i(TID) \Leftarrow$ get new candidate tuple from $HSE(i)$ and $RCTB$ 
\For{$ct=\{he_0,he_1,he_2,\cdots,he_{i-1},he_{i}\} \in CT_2(TID) $}
\State $B(i-1) \Leftarrow he_0\ XOR\ he_1$    
\State $B(i) \Leftarrow\ he_0\ XOR\ he_2$
\If{ left $q-\delta$ bits of $B(0)$ not equal to right $q-\delta$ bits of $B(1)$} 
\State Continue
\EndIf                                     
 \State insert $ct$ into $SCTB$
\EndFor
\end{algorithmic}
\end{algorithm}

When checking a candidate tuple newly adding a hot estimator in $HSE(i)$, only $B(i-1)$ and $B(i)$ should be tested. After update candidate tuple with the last row, $SCTB$ contains candidate tuple from which a valid host could be reconstructed. Set $Q=|SCTB|$, $U=\frac{Q}{V}$ and launch $V$ threads. Every thread scans $U$ or $U+1$ reconstructed hosts to estimate their opposite number according their union sliding estimators in the candidate tuple and check if they are sliding super points. By this method, every thread on GPU has the similar load with the cost of additional buffers for storing middle candidate tuples. Nowadays GPU has plenty global memory and the buffers not occupy many space because the number of sliding super points takes up a small part of hosts. 

$RSRA$ reconstructs sliding super points fast on GPU. But it can't get their cardinality and some fake host may hide in the candidate list. This problem will be solved by the method proposed in the next section.
\section{Memory efficient sliding cardinality estimation}
Linear estimator, $LE$, is a famous cardinality estimation algorithm\cite{DC:aLinearTimeProbabilisticCountingDatabaseApp}. It uses ${\eta}'$ bits, which are initialized to 0 at the beginning of a discrete time window, to estimate host's cardinality. When scanning a host $bip$ in $ST(aip,t,1)$, one bit in $LE$ selected by hash function $H_3$ will be set. $H_3(bip)$ maps $bip$ to a random value between 0 and ${\eta}'-1$. Let $|LE|$ represent the weight of $LE$, which means the number of 1 bit in it. At the end of a discrete time window, $|OH(aip,t,1)|$ will be estimated by the following equation.

\begin{equation}\label{eq_LE_cardinalityEst}
 \hat{|OH(aip,t,1)|}= - {\eta}'*ln(\frac{{\eta}'-|LE|}{{\eta}'})
\end{equation} 

But $LE$ only works when $k=1$. Like $SRE$, a sliding version of $LE$, $SLE$, is devised by replacing ${\eta}'$ bits in $LE$ with ${\eta}'$ short integers. Every short integer has the same operations as that in $SRE$. The $k$ weight of $SLE$ denoted as $|SLE|^k$ is the number of short integer whose value is smaller than $k$. $SLE$ estimates a host's cardinality by equation \ref{eq_SLE_cardinalityEst}.

\begin{equation}\label{eq_SLE_cardinalityEst}
 \hat{|OH(aip,t,1)|}= - {\eta}'*ln(\frac{{\eta}' - |SLE|^k}{{\eta}'})
\end{equation} 

To avoid allocating a $SLE$ for every host, a $SLE$ array, $SLEA$, is used to estimate all hosts' cardinalities. $SLEA$ contains ${r}'$ rows and every row has $2^{{q}'}$ $SLE$s. In order to have a high accuracy estimation of sliding super point's cardinality, ${\eta}'$ should be big enough. Generally, ${\eta}'$ should be no less than half of a sliding super point's cardinality \cite{DC:aLinearTimeProbabilisticCountingDatabaseApp}. But big ${\eta}'$ causes great memory consumption. To reduce the memory consumption of $SLEA$, two adjacent $SLE$s in a row share ${\eta}'-{\delta}'$ short integers. $SLEA$ could be regarded as an array of short integers with $r$ rows and $2^{{q}'}*{\delta}'+{\eta}'-{\delta}'$ columns as shown in figure \ref{fig_SLEA}.

\begin{figure}[!ht]
\centering
\includegraphics[width=0.47\textwidth]{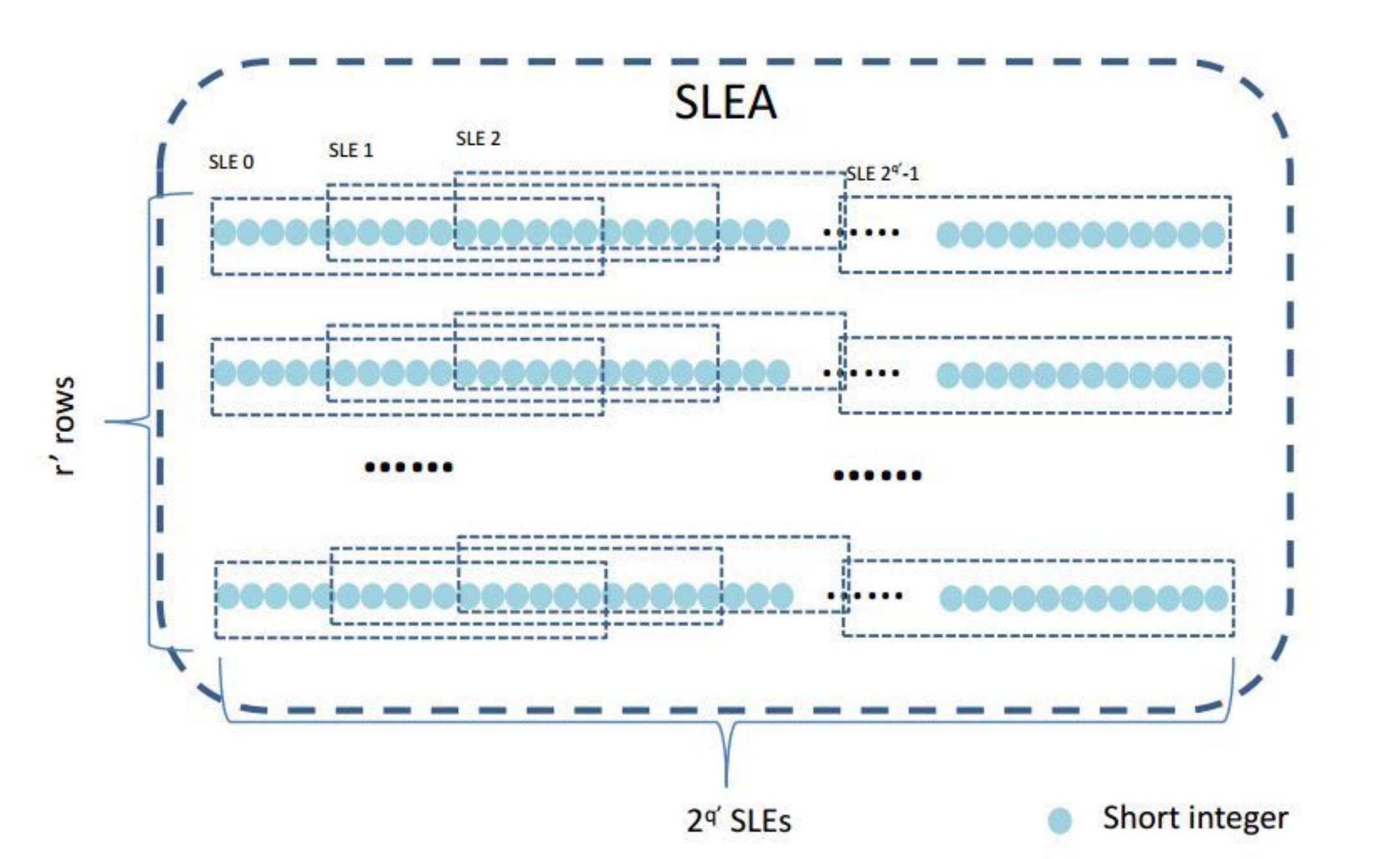}
\caption{Sliding linear estimator array}
\label{fig_SLEA}
\end{figure}

${\delta}'$ is the offset of short integers between two adjacent $SLE$s. When ${\delta}'$ is set to ${\eta}'$, no short integers will be shared by different $SLE$s in a row and there are total ${\eta}'*r*2^{{q}'}$ short integers. The memory reduction rate $MRR$ is defined below.
\begin{equation}\label{eq_SLEA_mrr}
 MRR({\delta}')=1- \frac{2^{{q}'}*{\delta}'+{\eta}'-{\delta}'}{{\eta}'*2^{{q}'}}
\end{equation}

Considering that $2^{{q}'}$ is much bigger than ${\eta}'-{\delta}'$, $MRR({\delta}')$ is determined by $\frac{{\delta}'}{{\eta}'}$. A small ${\delta}'$ will save memory greatly.
But ${\delta}'$ must be bigger than zero because when short integer offset is set to 0, there is only one $SLE$ in a row to record all hosts' cardinalities. For every IP pair $<aip,bip>$, $r'$ $SLE$s in each row will be selected by $r'$ random hash functions $LH_i(aip)$ where $LH_i$ maps $aip$ to a random value between 0 and $2^{{q}'}-1$. For every selecting $SLE$, a short integer determined by $H_3(bip)$ will be set to 0.

$SLEA$ could scan several IP pairs at the same time. High processing speed will be acquired if deployed it on GPU to run parallel. $SLEA$ locates on the global memory of GPU. When the IP pairs buffer is full, thousands of threads will be launched to process them at the same time. Each thread on GPU will run algorithm \ref{alg_SLEA_scan} to deal with a IP pair.

\begin{algorithm}                       
\caption{Sliding linear estimator array scans IP pair}          
\label{alg_SLEA_scan}
\begin{algorithmic}                    
\Require {\\  
Sliding linear estimator array $SLEA$,\\
IP pair $<aip,bip>$
\\}  
\State $sInt \Leftarrow H_3(bip)$
\For {$i \Leftarrow [0,r'-1]$}
\State $ \Leftarrow LH_i(aip)$
\State $SLEA[i,j][sInt] \Leftarrow 0$
\EndFor
\end{algorithmic}
\end{algorithm}
In algorithm \ref{alg_SLEA_scan}, $SLEA[i,j]$ points to the $j$th $SLE$ in the $i$th row.

Several GPU nodes could be used to scan different IP pairs in a distribute environment. Every node has a $SLEA$ with the same size: same rows, same columns number and same hash functions. When estimating hosts cardinality at the end of a time slice, $SLEA$ in different nodes should be merged together by algorithm \ref{alg_SLEA_merge}.

\begin{algorithm}                       
\caption{Sliding linear estimator array scans IP pair}          
\label{alg_SLEA_merge}
\begin{algorithmic}                    
\Require {\\  
Sliding linear estimator array set: \\
$SS=\{SLEA_0,SLEA_1,\cdots,SLEA_{n-1}\}$\\
Distributed nodes number $n$
\\}  
   \Ensure{Global sliding linear estimator array $GSLEA$}
\State Initialize $GSLEA$ 
\For{ $i \in [0, {r}'-1]$}
\For{ $ j \in [0, 2^{{q}'}*{\delta}'+{\eta}'-{\delta}'-1]$}
\State $v \Leftarrow 0$
\For {$ z \in [0,n-1]$}
\If {$ v< SLEA_z[i][j]$}
\State $v \Leftarrow SLEA_z[i][j]$
\EndIf
\State $GSLEA[i][j] \Leftarrow v$
\EndFor
\EndFor
\EndFor
\end{algorithmic}
\end{algorithm}

$SLEA_z[i][j]$ points to the $j$th short integer in the $i$th row of $SLEA$ in the $z$th node. Global $SLEA$ contains all hosts' cardinalities information and $GSLEA$ will be used to estimate cardinalities of candidate super points acquired by $SREA$. A $SLE$ in $SLEA$ will be shared by many hosts. For a certain host $aip$, there are ${r}'$ $SLE$s relating with it. Merging these $SLE$s to get a union one $USLE$ and estimating $|OH(aip,t,k)|$ from it could reduce the impact of other hosts. But in the $USLE$, some short integers would still be set by other hosts, especially when most of the short integers in $SLEA$ set to 0. 
\begin{definition}[SLEA row setting factor]
\label{def-SLEA_rsf}
For the $i$th row of $SLEA$, its setting factor $SF(i,k)$ is the ratio of the number of short integers in the $i$ row whose value is smaller than $k$ to $2^{{q}'}*{\delta}'+{\eta}'-{\delta}'$.
\end{definition}

$SF(i,k)$ reflects the usage of a row in $SLEA$. A big $SF(i,k)$ means that the $SLEA$ is used efficiently, but a host's $USLE$ will be effected by other hosts heavily. In order to remove this effect, the number of error setting short integers should be calculated. 

When $USLE$ is used by $aip$ exclusively, $|SLE|^k$ are expected to be $|SLE|^k={\eta}'- {\eta}'*e^{ - \frac{|OH(aip,t,k)|}{{\eta}'}}$ according to equation \ref{eq_SLE_cardinalityEst}. When $USLE$ contains short integers set by other hosts, the number of these error setting integers are expected to be $({\eta}'-|SLE|^k) * \prod_{i=0}^{{r}'-1}SF(i,k)$. Remove this value from $|USLE|^k$ and the rest value are expected to be $|SLE|^k$ as shown in equation \ref{eq_SLE_rmSInt_exptN}.
\begin{equation}\label{eq_SLE_rmSInt_exptN}
|SLE|^k = |USLE|^k - ({\eta}'-|SLE|^k) * \prod_{i=0}^{{r}'-1}SF(i,k)
\end{equation} 

The expectation of $|SLE|^k$ is acquired by modifying equation \ref{eq_SLE_modifiedWeight}.
\begin{equation}\label{eq_SLE_modifiedWeight}
\hat{|SLE|^k} =\frac{|USLE|^k-{\eta}'* \prod_{i=0}^{{r}'-1}SF(i,k)}{1-\prod_{i=0}^{{r}'-1}SF(i,k)}
\end{equation} 

Estimating $|OH(aip,t,k)|$ by $\hat{|SLE|^k}$ acquires a higher accuracy than using $|USLE|^k$ directly.

When estimating a host's cardinality, there are only reading operation to $GSLEA$. So several hosts' cardinalities could be estimated parallel in GPU. Algorithm \ref{alg_SLEA_CardinalityEST} describes how to estimate the cardinality of a given host from the $GSLEA$.

\begin{algorithm}                       
\caption{Sliding linear estimator array scans IP pair}          
\label{alg_SLEA_CardinalityEST}
\begin{algorithmic}                    
\Require {\\  
Global sliding linear estimator array $GSLEA$\\
Candidate sliding super point $aip$\\
}  
   \Ensure{Cardinality estimation ${|OH(aip,t,k)|}'$}
\State $j \Leftarrow LH_0(aip)$
\State $USLE \Leftarrow SLEA[0,j]$
\For{$i \in [1,{r}'-1]$}
\State $j \Leftarrow LH_i(aip)$
  \For {$z \in [0, {\eta}'-1]$}
     \If {$USLE[z] < SLEA[0][j*{\delta}'+z]$}
       \State $USLE[z] < SLEA[0][j*{\delta}'+z]$
     \EndIf
   \EndFor
\EndFor
\State $\hat{|SLE|^k} =\frac{|USLE|^k-{\eta}'* \prod_{i=0}^{{r}'-1}SF(i,k)}{1-\prod_{i=0}^{{r}'-1}SF(i,k)}$
\State $\hat{|OH(aip,t,k)|}' \Leftarrow - {\eta}'*ln(\frac{{\eta}' - |SLE|^k}{{\eta}'})$
\State Return $\hat{|OH(aip,t,k)|}'$
\end{algorithmic}
\end{algorithm}

Algorithm \ref{alg_SLEA_CardinalityEST} calculates the cardinality of every host in the candidate sliding super point list generated by $SREA$ and remove these hosts whose estimation is smaller than $\theta$ to improve the detection accuracy. Both $SREA$ and $SLEA$ could be updated parallel. With this algorithm, a common GPU is strong enough to detect sliding super points and estimate their cardinalities of the core network whose speed is as high as 40 Gb/s.
\section{Experiments and analyze}
We use a real world traffic to evaluate the performance of this sliding super point cardinality estimation algorithm SRLG. The traffic is OC192 downloading from Caida\cite{expdata:Caida}. This traffic contains one hour packets last from 13:00 to 14:00 on February 19, 2015. In our experiment, the threshold $\theta$ for super point is 1024. First we compare $SRLE$ with other algorithms under discrete time window. In the discrete time window, a time slice is set to 5 minutes. Under this time period, the one-hour traffic is divided into 12 sub traffics and we will detect super points from them. Table \ref{tbl-trafficInf} shows the detail information of every sub traffic. 
\begin{table*}
\centering
\caption{Traffic information}
\label{tbl-trafficInf}
\begin{tabular}{c}                                                                                                                                                                                                                           
\centering
\includegraphics[width=0.8\textwidth]{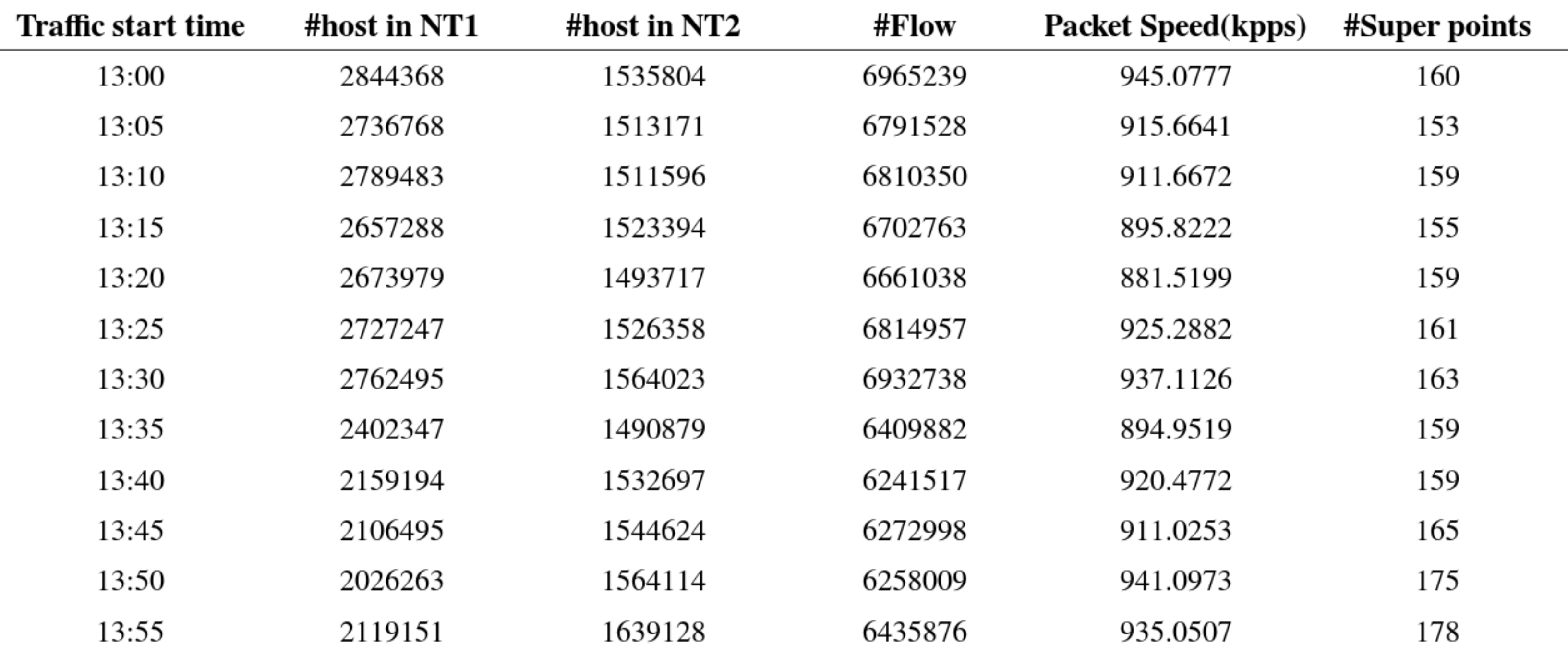}
\end{tabular}
\end{table*}

Accuracy, time consumption and memory requirement are three criteria to evaluate super point detection algorithm. False positive ratio FPR and False negative ratio FNR are two classic rates for accuracy comparing. They are given in definition \ref{def-fpr_fnr}.
\begin{definition}[FPR/FNR]
\label{def-fpr_fnr}
For a traffic with $N$ super points, an algorithm detects $N'$ super points. In the $N'$ detected super points, there are $N^+$ hosts which are not super points. And there are $N^-$ super points which are not detected by the algorithm. FPR means the ratio of $N^+$ to $N$ and FNR means the ratio of $N^-$ to $N$.
\end{definition}

FPR may decrease with the growth of FNR. If an algorithm reports more hosts as super point, its FNR will decrease but FPR will increase. So we use the sum of FPR and FNR, total false rate TFR, to evaluate the accuracy of an algorithm. 

To compare the performance of SRLG with other algorithms, we use DCDS\cite{HSD:ADataStreamingMethodMonitorHostConnectionDegreeHighSpeed}, VBFA\cite{HSD:DetectionSuperpointsVectorBloomFilter}, GSE \cite{HSD:GPU:2014:AGrandSpreadEstimatorUsingGPU} to compare with it. All of these algorithms are running on a common GPU card: GTX950 with 680 CUDA cores and 4 GB memory. The parameters of SRLG are: $\delta=5$, ${\delta}'=16$, ${\eta}'=2^{14}$, $\eta=8$, $q=q'=17$, $r=r'=5$.
Table \ref{tbl_avg_hsd_rlt} lists the average result of all the 12 sub traffics.

\begin{table*}
\centering
\caption{Average detection result}
\label{tbl_avg_hsd_rlt}
\begin{tabular}{c}                                                                                                                                                                                                                           
\centering
\includegraphics[width=0.7\textwidth]{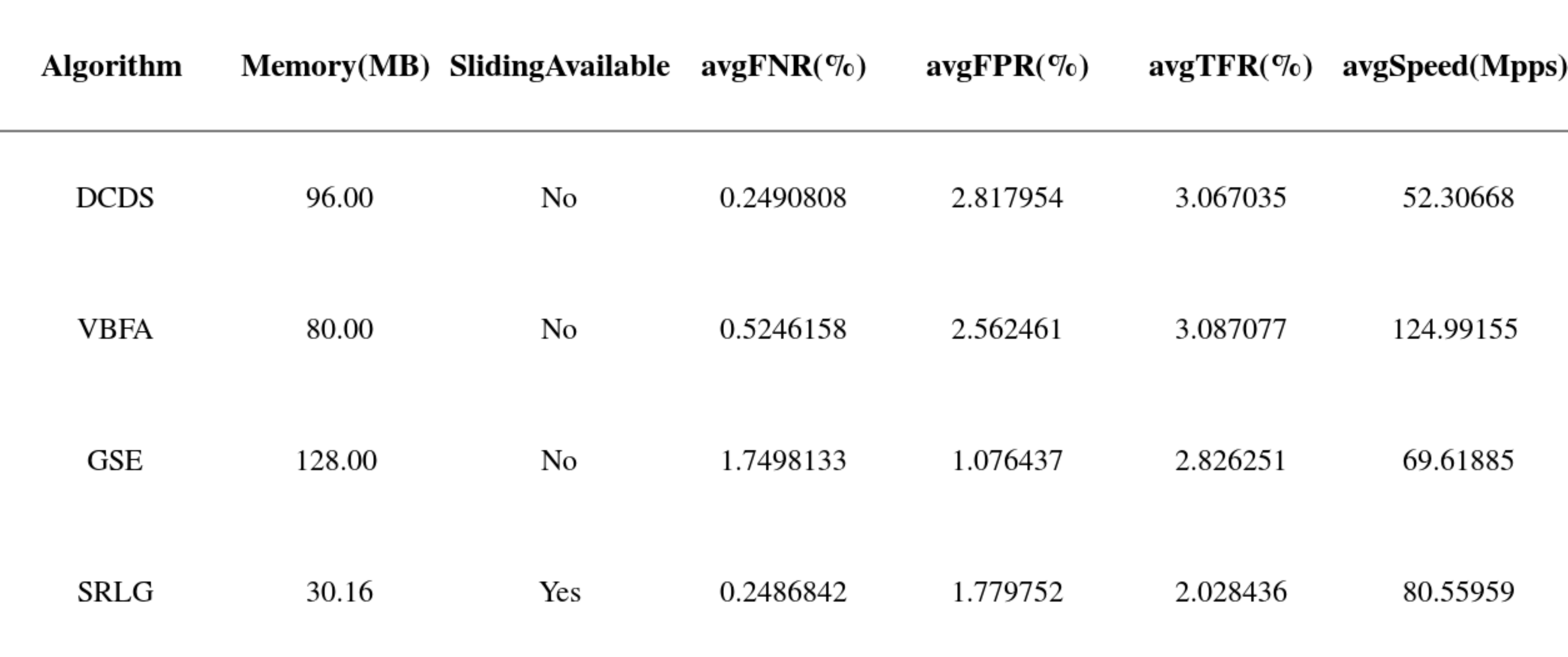}
\end{tabular}
\end{table*}

GSE has a lower FPR than other algorithms. It can remove fake super points according the estimating cardinality. But GSE may remove some super points too, which causes it has a higher FNR. Because it uses discrete bits to record host's cardinality, collecting all of these bits together when estimating super points cardinality will use lots of time. DCDS uses CRT when storing host's cardinality. CRT has a better randomness which makes DCDS has a lower FNR. But CRT is very complex containing many operations. So DCDS's speed is the lowest among all of these algorithms. VBFA has the fastest speed but its TFR is higher than that of SRLG.

From table \ref{tbl_avg_hsd_rlt} we can see that, SRLG uses the smallest memory, smaller than half of others' memory, and has the lowest total false rate. SRLG is the only one which can work under sliding time window. 

In the sliding time window experiments, a time slice is set to 1 second and $k$ is 300. SRLG's FPR, FNR and TFR are illustrated in figure \ref{fig_expSW_FPR}, \ref{fig_expSW_FNR} and \ref{fig_expSW_TFR}.
\begin{figure*}[!ht]
\centering
\includegraphics[width=0.8\textwidth]{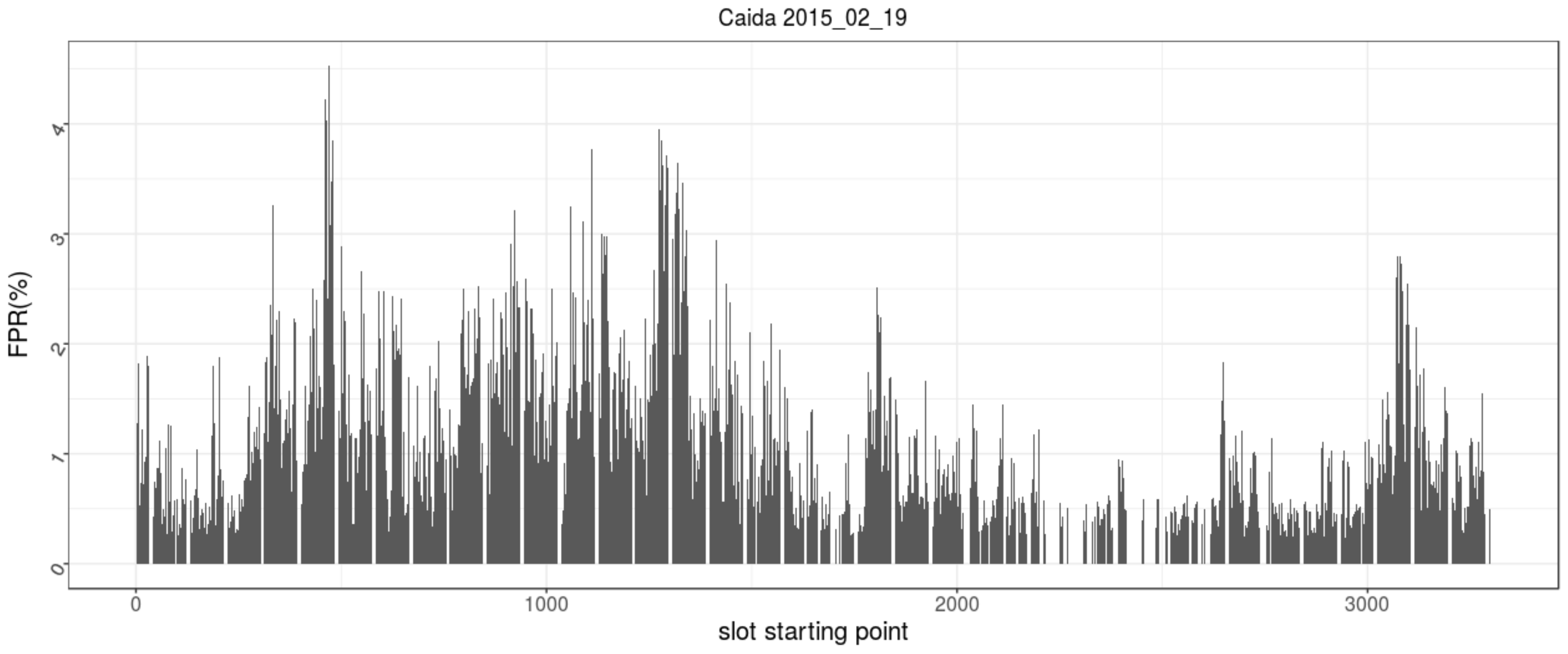}
\caption{Sliding time window FPR}
\label{fig_expSW_FPR}
\end{figure*}

\begin{figure*}[!ht]
\centering
\includegraphics[width=0.8\textwidth]{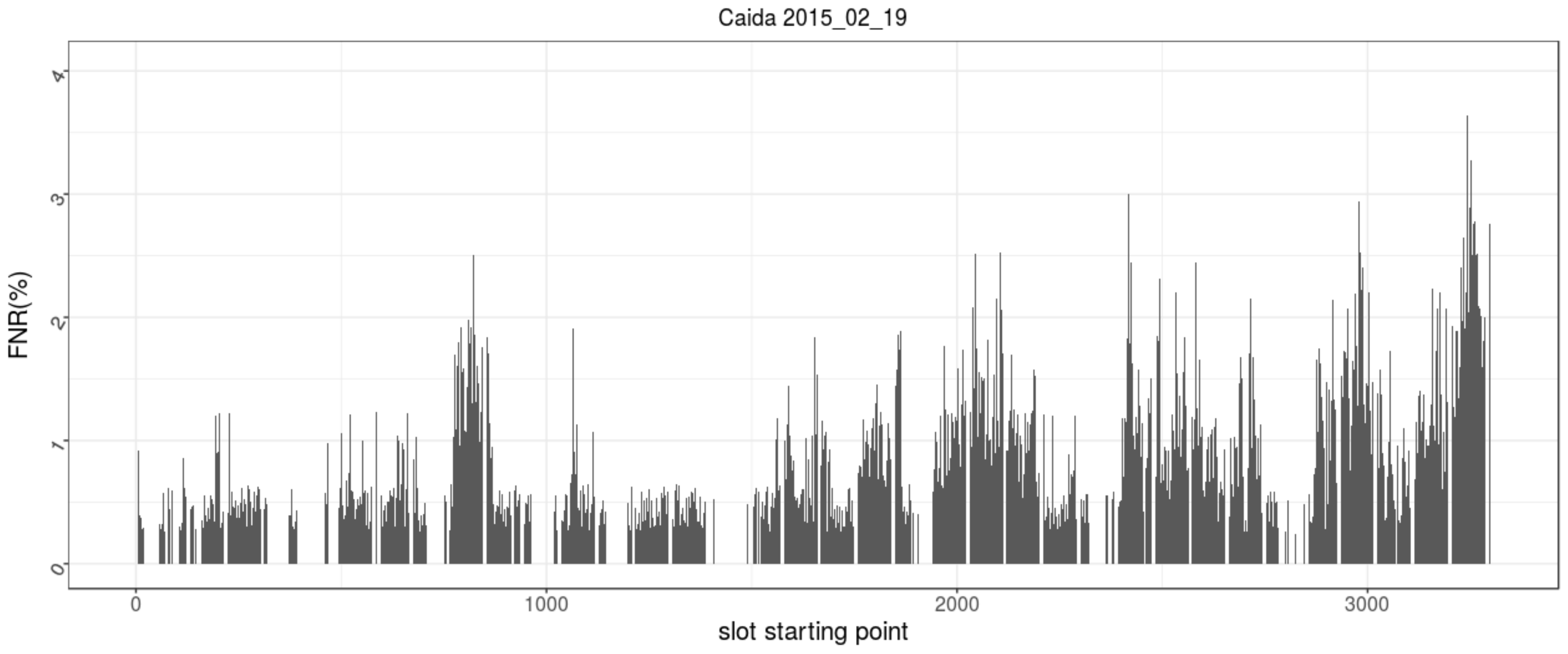}
\caption{Sliding time window FNR}
\label{fig_expSW_FNR}
\end{figure*}

\begin{figure*}[!ht]
\centering
\includegraphics[width=0.8\textwidth]{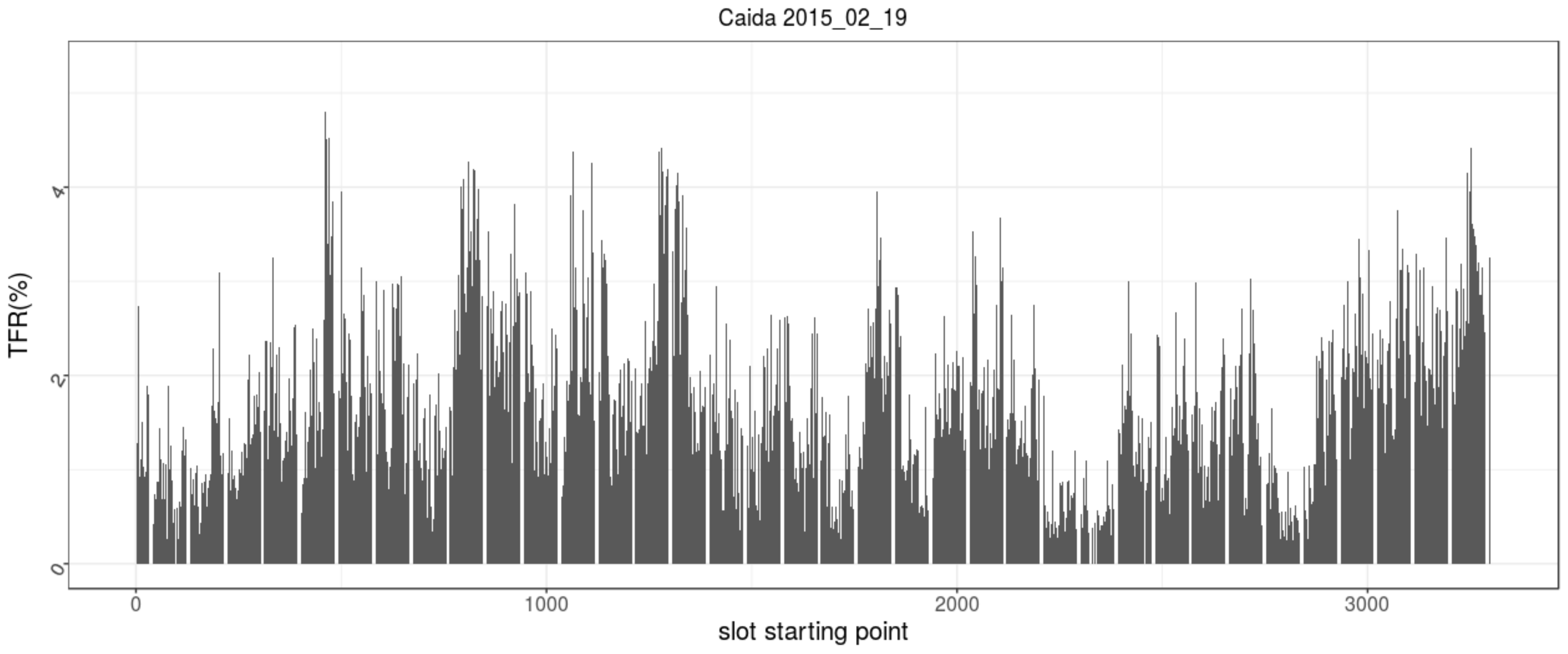}
\caption{Sliding time window TFR}
\label{fig_expSW_TFR}
\end{figure*}

Under most sliding time window, SRLG has a low FPR, smaller than 1.5\%. When FNR is small, FPR is relative high. But the total false rate is stably small. This experiments show that SRLG has low TFR and smallest memory for sliding super point detection in core network. It can be applied to a bigger network by increasing of the size of $SREA$ and $SLEA$.

\section{Conclusion}
 Sliding super point cardinality estimation is an important issue in network research areas. This paper firstly proposed an algorithm SRLG to solve this problem. SRLG has the ability to run parallel in distributing environment. It uses two novel cardinality estimation methods: SRE and SLE. Based on SRE and SLE, two smart structures $SREA$ and $SLEA$ are devised. SREA detects sliding super points and generates a candidate list by the novel reversible hash functions $RHFG$. SLEA estimates the cardinality of every host in the candidate list with high accuracy. By sharing short integers of different SLEs, SLEA consumes very small memory, even smaller than those algorithms running under discrete time window. Small memory consumption reduces the communication cost between different nodes which is always the bottle neck of many distributing algorithm. Both SREA and SLEA could be updated parallel. When deployed on GPU, SRLG can deal with high speed network in real time.

\iftoggle{ACM}{
\bibliographystyle{unsrt}
}

\bibliography{..//..//ref} 

\end{document}